\documentclass{article}
\usepackage{amsmath,amsthm}
\usepackage{amssymb}

\newtheorem{theorem}{Theorem}[section]

\newtheorem{lemma}[theorem]{Lemma}
\newtheorem{corollary}[theorem]{Corollary}

\theoremstyle{definition}

\newtheorem{remark}[theorem]{Remark}
\newtheorem{assumption}[theorem]{Assumption}

\def\Mc{{\cal M}}

\def\Fc{{\cal F}}

\def \a{\alpha}
\def \b{\beta}
\def \g{\gamma}

\def \l{\lambda}

\def \e{\varepsilon}
\newcommand\R{{\mathbb{R}}}

\begin{document}

\title{From small markets to big markets}

\author{Laurence Carassus
\thanks{L\'{e}onard de Vinci P\^ole Universitaire, Research Center, 92 916 Paris La D\'{e}fense, France and 
LMR, FRE 2011 Universit\'{e} de Reims-Champagne Ardenne, France.
E-mail: laurence.carassus@devinci.fr}
 \and Mikl\'os R\'asonyi\thanks{Alfr\'ed R\'enyi Institute of
Mathematics, Budapest, Hungary.
E-mail: rasonyi@renyi.mta.hu}}

\maketitle

\begin{abstract}
We study the most famous example of a large financial market: the Arbitrage Pricing Model, where investors can trade
in a one-period setting with countably many assets admitting a factor structure. We consider the problem of maximising expected utility in this
setting. Besides establishing the existence of optimizers under weaker assumptions than previous papers,
we go on studying the relationship
between optimal investments in finite market segments and those in the whole market. We show that certain natural
(but nontrivial) continuity rules hold: maximal satisfaction, reservation prices and (convex combinations of) optimizers computed
in small markets converge to their respective counterparts in the big market.
\end{abstract}

Keywords: {Arbitrage Pricing Theory, Large markets, Maximisation of expected utility.}

MSC classification:{Primary 93E20, 91B70, 91B16; Secondary 91G10, 46B09.}

\section{Introduction}

Arbitrage Pricing Theory (APT) was conceived by \cite{ross} in order to derive the conclusions of Capital Asset Pricing Model
(see \cite{lintner,sharpe})
from alternative assumptions. These remarkable conclusions had a huge bearing on empirical work
but they somehow overshadowed the highly inventive model suggested in \cite{ross}. 

Mathematical
finance subsequently took up the idea of a market with countably many assets and
the theory of large financial markets was founded in \cite{kk1} and further
developed in e.g.\ \cite{kk2,iw1,iw2,klein,josef}, just to mention a few. For the sake
of generality, continuous trading was assumed in the overwhelming majority of related papers
which, again, eclipsed the original setting of \cite{ross}. 

While the arbitrage theory of the large financial markets has been worked out in \cite{kk1,kk2}
satisfactorily in continuous time, other crucial topics -- such as utility maximization or superreplication --
brought about only dubious conclusions and unsettled questions. Portfolios
in finitely many assets were considered in the above references and a natural definition for strategies
involving possibly all the assets was missing. Generalized portfolios were introduced
(see \cite{paolo,josef,oleksii}) as suitable limits of portfolios with finitely many
assets. They lacked, however, a clear economic interpretation. In the APT (and, for the moment,
only in that model)  \cite{jmaa} introduces a straightforward concept of portfolios
in infinitely many assets which we will use in the present paper. In \cite{CR18} it is proved that 
assuming absence of arbitrage in all of the small markets and under integrability conditions, the no arbitrage 
condition stated with infinitely many assets also holds true. In the same paper, the authors obtain a  
dual representation of the superreplication cost of a contingent claim. 

In this paper, we investigate  the existence of optimizers for utility functions  on the whole real line (the  positive
real axis case was treated in \cite{CR18}) and we relax some rather stringent conditions imposed in \cite{jmaa,ijtaf}. 
From both a theoretical and a computational viewpoint it is crucial to clarify the relationship
between optimal investment in the finite markets and those in the whole market. 

In our setup, 
it is expected that the value functions in finite markets perform asymptotically as well
as the value function in the large market. 
Considering utility indifference prices, these should also
converge as the number of assets increases. While these facts are intuitive, no formal
justification has been provided so far. We prove these facts in Theorem \ref{zut} and Corollary \ref{prixut} below. 
We also prove that certain convex combinations of the optimal portfolios in finite markets perform asymptotically as well
as the overall optimizer. 

Asymptotic results for superhedging and mean-variance hedging have been obtained in \cite{baran,campi}. In the
utility maximization context the first such result is Theorem 5.3 of \cite{jmaa} where it was shown
that there exists a sequence of strategies in finite markets whose values converge to the optimal
value. That paper, however, assumed
that asset price changes may take arbitrarily large negative and positive values which is a rather strong requirement.
Under the more relaxed conditions of the present work we also show the existence of such sequence, moreover,
they can be chosen to be averages of finite market optimizers, see Theorem \ref{matural} below.

Section \ref{lmm} presents the model and recalls some useful results from \cite{CR18}. 
Section \ref{secut} contains the main contributions: existence of utility maximization and the asymptotics from
small markets to big markets.

\section{The large market model}
\label{lmm}

Let $({\Omega}, \Fc, P)$ be a probability space. We consider a two stage Arbitrage Pricing Model.
For any $i \geq 1$, let the return on asset $i$ be given by
\begin{eqnarray*}
R_i &=& \bar{\beta}_i(\varepsilon_i-b_i),\quad 1\leq i\leq m;\\
R_i &=& \sum_{j=1}^m\beta_i^j(\varepsilon_j-b_j)+\bar{\beta}_i(\varepsilon_i-
b_i),\quad i>m,
\end{eqnarray*}
where the $(\varepsilon_i)_{i \geq 1}$ are random variables and $(\bar{\beta}_i)_{i \geq 1}, (b_i)_{i \geq 1},(\beta_i^j)_{i >m, 1 \leq j \leq m}$ are
constants. We refer to \cite{kk1,def, ijtaf} for further discussions on the model.
\begin{assumption}
\label{un}
The $(\e_i)_{i \geq 1}$ are square-integrable, independent random variables satisfying
\begin{eqnarray*}
E(\varepsilon_i)=0,\quad E\left(\varepsilon_i^2\right)=1,\quad i\geq 1.
\end{eqnarray*}
\end{assumption}
We consider
strategies using potentially  infinitely many assets and  belonging  to
$$\ell_2:=\left\{(h_i)_{i\geq 1}, \, h_{i}\in\mathbb{R},\, i\geq 1,\,  \sum_{i=1}^{\infty}h_i^2<\infty\right\}$$ which is an Hilbert space with the norm $||h||_{\ell_2}:=\sqrt{\sum_{i=1}^{\infty}h_i^2}$. \\
Let $L^2(\Omega, \Fc,P):=\{X:\Omega \to \R, \, E|X|^2< \infty\}$ (denoted by $L^2(P)$ from now on), which is again a Hilbert space with the norm $||X||_{L^2}:=\sqrt{E(|X|^2)}$.
For  $h \in \ell_2,$  let $\Phi(h):=\sum_{i=1}^{\infty}h_i\e_i,$ where 
the infinite sum in $\Phi(h)$ has to be understood as the limit in $L^2(P)$ of the finite sequences $(\sum_{i=1}^{n}h_i\e_i)_{n \geq 1}$.
Then $\Phi$ is an isometry from $\ell_2$ to $L^2(P).$
\begin{assumption}\label{b}
We have $\|b\|_{\ell_2}<\infty$.
\end{assumption}
Under Assumption \ref{b}, we have (see (5) in \cite{CR18}) that 
\begin{equation}
\label{isol}
E\left(\left(\sum_{i=1}^{\infty}h_i(\e_i-b_i)\right)^2\right) \leq  (1+\|b\|_{\ell_2}^2) \| h\|^2_{\ell_2}<\infty,
 \end{equation}
and we may consider again the infinite sum $\langle h, \e-b\rangle:=\sum_{i=1}^{\infty}h_i(\e_i -b_i)$. Note that 
$$E(|\langle h, \e-b\rangle |) \leq\sqrt{E\left(\langle h, \e-b\rangle \right)^2}\leq  \sqrt{1+\|b\|_{\ell_2}^2} \| h\|_{\ell_2}.$$
The (self-financed) value at time $1$ that can be attained starting from $x$  and using a strategy $h$ in $\ell_2$  
with infinitely many assets is given by 
$$
V^{x,h}:= x+\langle h, \e-b\rangle.{}
$$

\begin{assumption}
\label{AOAfini}
For all $i \geq 1$,
$$P(\e_i >b_i)>0 \mbox{ and } P(\e_i <b_i)>0.$$
\end{assumption}
Fix $N\geq 1$. Using Lemma 3.4 in \cite{CR18}, under Assumptions \ref{un} and \ref{AOAfini}, there exists some
${\alpha}_{N}\in (0,1)$ such that for every $(h_1,\ldots, h_N) \in \R^N$ satisfying $\sum_{i=1}^N h_i^2=1$ we have 
\begin{eqnarray}
\label{aoapetitNalpha}
P\left(\sum_{i=1}^N h_i(\e_i - b_i)<-{\alpha}_{N}\right)>{\alpha}_{N}.
\end{eqnarray}
This condition is the so called quantitative no-arbitrage condition on any ``small market''  with $N$ 
random sources and it is well-known that this condition is equivalent to the existence of a equivalent martingale measure
for the finite market with assets $R_{1},\ldots,R_{N}$ (see \cite{dmw} and \cite{follmer-schied}). 

However, we need the existence of martingale measures for the whole market and even sufficient
integrability of the martingale density. We say that
EMM2 holds true if
\begin{eqnarray}
\label{mmset}
\Mc_ 2:=\left\{Q \sim P, \, \frac{dQ}{dP} \in L ^2(P), \, E_{Q} (\e_i)=b_i, \, \forall i\geq 1\right\}\neq \emptyset.
\end{eqnarray}

Unfortunately, Assumptions \ref{un}, \ref{b} and \ref{AOAfini} are not known to be sufficient for ensuring that EMM2  
holds true (see Proposition 4 of \cite{def}). Hence we also need the following technical condition.

\begin{assumption}\label{trois}
We have that
\begin{equation}\label{harom}
\sup_{i\geq 1} E\left[|\varepsilon_i|^3\right]<\infty.
\end{equation}
\end{assumption}

\begin{lemma} Under
Assumptions \ref{un}, \ref{AOAfini} and \ref{trois},
\begin{equation}
\label{eqaaa}
\mbox{ Assumption}  \, \ref{b}
\Longleftrightarrow \mbox{ EMM2.}
\end{equation}
\end{lemma}
\begin{proof}
This is Corollary 1 of \cite{def}.
\end{proof}

Lemma \ref{aoaquant} below  asserts that the 
quantitative no arbitrage condition, mentioned above, holds true in the large market, too. 

\begin{lemma}
\label{aoaquant}
Assume that Assumptions \ref{un}, \ref{b}, \ref{AOAfini} and \ref{trois} hold true.
Then there exists some $\a>0$, such that  for all $h \in \ell_2$ satisfying $\|h\|_{\ell_2}=1$
 $$P(\langle h,\e\rangle <-\a) > \a.$$
\end{lemma}
\begin{proof}
This is Proposition 3.14 in \cite{CR18}.
\end{proof}
\begin{remark}
\label{nulla} {\rm
If $Q\in\mathcal{M}_2$ is such that $dQ/dP\in L^2$ and if Assumption \ref{b} holds true then
$E_Q\left(V^{0,h}\right)=0$ for all $h \in \ell_2$, see Remark 3.4  of \cite{CR18}.}
\end{remark}

Lemma \ref{miki} below will be used 
in the proofs of Theorems \ref{csonti} and \ref{zut} in order to show uniform integrability.   
\begin{lemma}
\label{miki}
Assume that Assumptions \ref{un} and  \ref{b}  hold true and that 
$\sup_{i \geq 1} E|\varepsilon_i|^{\gamma}<\infty$ for some $\gamma\geq 2$. Then there is a constant $C_{\gamma}$
such that, for all $h\in\ell_2$
$$
E|\langle h,\varepsilon-b\rangle|^{\gamma}\leq C_{\gamma}\Vert h\Vert_{\ell_2}^{\gamma}(1+\Vert b\Vert_{\ell_2}^{\gamma}).
$$
\end{lemma}
\begin{proof}
This is Lemma 3.10 in \cite{CR18}.
\end{proof}
\begin{remark}
\label{remui}
Let $0<\l < \g$ and $c>0$. Fix $h \in \ell_2, \, \|h\|_{\ell_2}\leq c$. Using Assumption \ref{trois}, Holder inequality and Lemma \ref{miki}, we get that 
for any $A  \in \Fc$,
\begin{eqnarray*}
E(|V^{x,h}|^{\l} 1_A) & \leq & 2^{\l -1} |x|^{\l} P(A) + 2^{\l -1}E(|\langle h,\varepsilon-b\rangle|^{\l} 1_A) \\
& \leq & 2^{\l -1} |x|^{\l} P(A) + 2^{\l -1}(E(|\langle h,\varepsilon-b\rangle|^{\g}))^{\l/\g} (P(A))^{1/q} \\
& \leq & 2^{\l -1} |x|^{\l} P(A) + 2^{\l -1}c^\l(C_{\gamma}(1+\Vert b\Vert_{\ell_2}^{\gamma}))^{\l/\g} (P(A))^{1/q},
\end{eqnarray*}
where $q$ is the conjugate of $\g/\l$. 
So an important consequence of Assumption \ref{trois} is that for any $c>0$ and $0<\l <3$ 
$\{|V^{x,h}|^{\l}, \,h \in \ell_2, \, \|h\|_{\ell_2}\leq c \}$ is  uniformly integrable.
\end{remark}
We finally recall an important concept of functional analysis. A Banach space $\mathbb{B}$
has the \emph{Banach-Saks property} if, for every norm-bounded sequence $\xi_{n}\in\mathbb{B}$, $n\in\mathbb{N}$,
there exists a subsequence $n_{k}$, $k\in\mathbb{N}$ such that the corresponding arithmetic means
$$
\frac{\xi_{n_{0}}+\ldots+\xi_{n_{k-1}}}{k}
$$
converge in the norm of $\mathbb{B}$. It was proved in \cite{banach-saks} that $L_{p}$
spaces have this property. In the present paper we will apply this result in
the Hilbert space $\ell_{2}$.

\section{Utility maximisation}
\label{secut}

It is standard (see \cite{von}) to model economic agents' preferences by
concave increasing utility functions $U$. So suppose that $U:\mathbb{R}\to\mathbb{R}$ is a  concave
strictly increasing   differentiable function and that  for some $x_0 \in \R$
\begin{eqnarray}
\label{norma}
U(x_0)=0 \mbox{ and } U'(x_0)=1.
\end{eqnarray}
For a claim $G \in  L^0$  and $x \in \mathbb{R} $, we define
$$\mathcal{A}(U,G,x):=\left\{ h \in \ell_2,\;E U^{-}(V^{x,h}-G)<+\infty \right\}.$$
Define  the supremum of expected utility at the terminal date when delivering a contingent claim $G$, starting
from initial wealth $x \in \R$, by
\begin{eqnarray}
\label{gnon}
u(G,x):=\sup_{h \in\mathcal{A}(U,G,x)}EU(V^{x,h}-G).
\end{eqnarray}

The following assumptions will be needed in Theorems \ref{csonti} and \ref{zut}.
\begin{assumption}
\label{queuneg1}
There exists some constants $C_1 \in (0,\infty)$, $C_2 \in \mathbb{R}_{+}$
and $\beta>1$ such that for all $x \leq x_0$
$$|U(x)|\geq C_1|x|^{\beta} -C_2.$$
\end{assumption}

\begin{assumption}
\label{queuneg2}
There exists some constants $C_3 \in (0,\infty)$, $C_4 \in\mathbb{R}_{+}$ and $\gamma\geq \max(\beta,2)$ such that for all $x \in \R$
$$U^-(x)\leq C_3|x|^{\gamma} +C_4$$
and
\begin{equation}\label{harom2}
\sup_{i\geq 1} E\left[|\varepsilon_i|^{\gamma}\right]<\infty.
\end{equation}
\end{assumption}
\begin{assumption}
\label{intG}
We have $G \geq 0$ a.s.\ and it satisfies
$|E(U(x-G))|<+\infty$, for all $x \in \R$. \end{assumption}

\begin{remark} Assumption \ref{intG} is satisfied whenever $G$ is nonnegative, measurable and bounded.
Define
\[
U(x):=-\frac{1}{\delta}[ (x+1)^{-\delta}-1]
1_{\{x>0\}} -\frac{1}{\beta}[ (1-x)^{\beta}-1]1_{\{x\leq 0\}}
\]
for some $\beta\geq 2$ and $\delta>0$. Then $U$ is concave, strictly increasing, continuously differentiable and
it satisfies both Assumptions \ref{queuneg1}
and \ref{queuneg2} whenever $\sup_{i\geq 1} E\left[|\varepsilon_i|^{\beta}\right]<\infty$.	
Note that Assumption \ref{trois} implies \eqref{harom2} when $2 \leq \beta\leq 3$.
\end{remark}

\begin{remark} Let $U$ be concave, strictly increasing and differentiable, satisfying Assumptions \ref{queuneg1}, \ref{queuneg2}
and \ref{intG}. Then \eqref{norma} actually imposes no restriction on $U$.
Indeed, as $U$ cannot be constant, there exists $x_0 \in \R$ such that $U'(x_0)>0$. Define
$$V(x):=\frac{U(x)}{U'(x_0)}-\frac{U(x_0)}{U'(x_0)},$$ 
which
obviuosly satisfies \eqref{norma}. Moreover,
\begin{eqnarray*}
|V(x)| 
 & \geq & \frac{C_1}{U'(x_0)} |x|^{\beta} -\frac{C_2}{U'(x_0)} -\frac{|U(x_0)|}{U'(x_0)}, \; \; x \leq x_0\\
V^-(x)& \leq & \frac{C_3}{U'(x_0)} |x|^{\gamma} +\frac{C_4}{U'(x_0)} +\frac{U^+(x_0)}{U'(x_0)}, \; \; x \in \R \\
|E(V(x-G))| & \leq &  \frac{|E(U(x-G))|}{U'(x_0)} +\frac{|U(x_0)|}{U'(x_0)}<\infty. 
\end{eqnarray*}
So Assumptions \ref{queuneg1}, \ref{queuneg2} and \ref{intG} hold true for $V$. 
One may apply Theorems \ref{csonti}, \ref{zut} and Corollary \ref{prixut} below to $V$ and then these same results can be deduced for $U$, too. 
\end{remark}

The following lemmata will be used in the proofs of Theorems \ref{csonti} and \ref{zut}. 

\begin{lemma}
\label{toutva}  Let Assumption \ref{b} hold true and assume
$G \geq 0$ a.s. Then for all $y \in \R$ and $h \in \ell_2$
\begin{eqnarray}
\label{ilfaitbeau}
U^+(y+\langle h, \e-b\rangle-G)
& \leq & |x_0|+  |y+\langle h, \e-b\rangle|.
\end{eqnarray}
\end{lemma}
\begin{proof}
As $U$ is increasing, concave and differentiable, recalling \eqref{norma}, we get  for all $y \in \R$,
\begin{eqnarray*}
U(y) & \leq & U(\max(x_0,y))\leq U(x_0)+ \max(y-x_0,0) U\rq{}(x_0)\\
 & \leq & \max(y-x_0,0)\leq |y-x_0|\leq |y| +  |x_0|.
\end{eqnarray*}
Let $h \in \ell_2$, we get that
\begin{eqnarray}
\nonumber
& & U^+(y+\langle h, \e-b\rangle-G)
\leq    U^+(y+\langle h, \e-b\rangle) \\
\nonumber
&  \leq &  U^+(y+\langle h, \e-b\rangle) 1_{y+\langle h, \e-b\rangle \geq x_0} +  U^+(x_0) 1_{y+\langle h, \e-b\rangle <x_0} \\ 
\nonumber
& = &  U(y+\langle h, \e-b\rangle) 1_{y+\langle h, \e-b\rangle \geq x_0}   \leq |x_0|+|y+\langle h, \e-b\rangle|.
\end{eqnarray}
\end{proof}

Lemma \ref{hborne} asserts that an optimal solution for \eqref{gnon} must be bounded. 
 \begin{lemma}
\label{hborne}
Assume that  Assumptions \ref{un}, \ref{b},  \ref{AOAfini}, \ref{trois}, \ref{queuneg1} and \ref{intG}  hold true.
Let $x \in \R.$  There exists some constant $M_{x,G} >0$ such that if $h \in \ell_2$ satisfies 
$$\|h\|_{\ell_2}> M_{x,G}$$
then the $0$ strategy performs better than $h$, that is,
$$
EU(x-G) > EU(x+\langle h,\e-b\rangle-G).
$$
\end{lemma}

\begin{proof}
Let $x \in \R$ and $h \in \ell_2$. 
Recall $\a>0$ from Lemma \ref{aoaquant}. 
As $b \in \ell_2$, there exists some $n_{\a}\geq 1$ such that $\left(\sum_{i \geq n_{\a}+1} b_i^2\right)^{1/2} \leq \alpha/2$. 
Let 
\begin{eqnarray*}
\underline h:=(h_1,\ldots,h_{n_{\a}},0, \ldots,) &\mbox{ and }& \underline b=:(b_1,\ldots,b_{n_{\a}}, 0,\ldots,)\\
\overline h:=(0,\ldots,0,h_{n_{\a}+1}, \ldots,) & \mbox{ and }& \overline b=:(0,\ldots,0,b_{n_{\a}+1}, \ldots,).
\end{eqnarray*}
From the no-arbitrage condition in the market with ${n_{\a}} $  assets (see  \eqref{aoapetitNalpha}) there exits $\a_{n_{\a}}$ such that $P(A)>\a_{n_{\a}},$ where 
$A:=\left\{\sum_{i=1}^{n_{\a}} h_i(\e_i-b_i)<-\a_{n_{\a}}\|\underline h\|_{\ell_2}\right\}$. 
Let 
$B:=\left\{\sum_{i \geq n_{\a}+1} h_i\e_i\leq -\a \|\overline h\|_{\ell_2}  \right\}$ then $P(B)>\a$ (recall Lemma \ref{aoaquant}). 
As the $(\e_i)_{i\geq 1}$ are independent, we get that
$P(A\cap B)=P(A)P(B)>\a_{n_{\a}} \a$. 
On $A\cap B$,
 \begin{eqnarray*}
\langle h,\e-b\rangle  & = & \langle\underline h,\e-b\rangle + 
\langle\overline h,\e-b\rangle 
 \leq   -\a_{n_{\a}}\|\underline h\|_{\ell_2} -\a \|\overline h\|_{\ell_2} 
- \langle\overline h,\overline b\rangle  \\
 & \leq &  -\a_{n_{\a}}\|\underline h\|_{\ell_2} -\a \|\overline h\|_{\ell_2} +
 \|\overline b\|_{\ell_2}\|\overline h\|_{\ell_2}\\
 & \leq &  -\a_{n_{\a}}\|\underline h\|_{\ell_2} -\a \|\overline h\|_{\ell_2} +
 \a/2\|\overline h\|_{\ell_2} \leq  -\overline{\a} (\|\underline h\|_{\ell_2}+\|\overline h\|_{\ell_2}),
 \end{eqnarray*}
where $\overline{\a}=\inf(\a_{n_{\a}},\a/2)$. Thus 
 $P(\langle h,\e-b\rangle <-\overline{\a} (\|\underline h\|_{\ell_2}+\|\overline h\|_{\ell_2}))>\a_{n_{\a}} \a$. 
Assume that  $\|\underline h\|_{\ell_2}+\|\overline h\|_{\ell_2} \geq \max \left( \frac{x-x_0}{\overline{\a}},\frac{|x|}{\overline{\a}}\right)$. Then applying Lemma \ref{toutva} and Assumption \ref{queuneg1}, we get that
 \begin{eqnarray*}
EU(V^{x,h}-G) & \leq  &  EU(x+\langle h,\e-b\rangle) \\
& \leq  &
E\left(U(x+\langle h,\e-b\rangle)1_{\langle h,\e-b\rangle <-\overline{\a} (\|\underline h\|_{\ell_2}+\|\overline h\|_{\ell_2})}\right) +\\
 & &
 E\left(U^+(x+\langle h,\e-b\rangle)1_{\langle h,\e-b\rangle \geq -\overline{\a} (\|\underline h\|_{\ell_2}+\|\overline h\|_{\ell_2})}\right)\\
  & \leq & U(x-\overline{\a} (\|\underline h\|_{\ell_2}+\|\overline h\|_{\ell_2})) \a_{n_{\a}} \a + |x_0|+ E|x+\langle\underline h,\e-b\rangle + 
\langle\overline h,\e-b\rangle|\\
  & \leq & U(x-\overline{\a} (\|\underline h\|_{\ell_2}+\|\overline h\|_{\ell_2})) \a_{n_{\a}} \a + |x_0|+|x|+\|\underline h\|_{\ell_2}\sqrt{1+ \|\underline b\|^2_{\ell_2}}\\
  & & 
  +\|\overline h\|_{\ell_2}\sqrt{1+ \|\overline b\|^2_{\ell_2}}
  \\
  & \leq & \left(-C_1 \left|\overline{\a} (\|\underline h\|_{\ell_2}+\|\overline h\|_{\ell_2})-x\right|^{\beta}+ C_2 \right)\a_{n_{\a}} \a + |x_0|+|x|\\
   & &+(\|\underline h\|_{\ell_2} +\|\overline h\|_{\ell_2})\sqrt{1+ \|b\|^2_{\ell_2}}\\
  & \leq & \left(-C_1 \overline{\a} ^{\beta} (\|\underline h\|_{\ell_2}+\|\overline h\|_{\ell_2})^{\beta}+ C_2 \right)\a_{n_{\a}} \a + |x_0|+|x|\\
   & &(\|\underline h\|_{\ell_2}+\|\overline h\|_{\ell_2})\sqrt{1+ \|b\|^2_{\ell_2}}.
 \end{eqnarray*}
because $U(x-\overline{\a} (\|\underline h\|_{\ell_2}+\|\overline h\|_{\ell_2})) \leq U(x_0) =0$ and
\begin{eqnarray*}\left|\overline{\a} (\|\underline h\|_{\ell_2}+\|\overline h\|_{\ell_2})-x\right|^{\beta} &\geq& \left|\overline{\a} (\|\underline h\|_{\ell_2}+\|\overline h\|_{\ell_2})-|x|\right|^{\beta}
=\left(\overline{\a} (\|\underline h\|_{\ell_2}+\|\overline h\|_{\ell_2})-|x|\right)^{\beta} \\
&\geq &\overline{\a} ^{\beta} (\|\underline h\|_{\ell_2}+\|\overline h\|_{\ell_2})^{\beta}.
\end{eqnarray*}
Assume that
 \begin{eqnarray*}
(\|\underline h\|_{\ell_2}+\|\overline h\|_{\ell_2})\sqrt{1+ \|b\|^2_{\ell_2}} -\frac{C_1}2 \a_{n_{\a}} \a\overline{\a}^{\beta}(\|\underline h\|_{\ell_2}+\|\overline h\|_{\ell_2})^{\beta}  <  0\\
-\frac{C_1}2 \overline{\a}^{\beta}\a_{n_{\a}} \a (\|\underline h\|_{\ell_2}+\|\overline h\|_{\ell_2})^{\beta} + |x_0|  + |x| + C_2\a_{n_{\a}} \a  < -|EU(x-G)| \leq EU(x-G),
 \end{eqnarray*}
which is true if $\|\underline h\|_{\ell_2}+\|\overline h\|_{\ell_2}> \overline M_{x,G},$ where 
\begin{eqnarray*}
 \overline M_{x,G}  &:= &
\max \left(\left(2\frac{|x_0| +|x| + C_2\a_{n_{\a}} \a +|E(U(x-G))|}{{C_1} \a_{n_{\a}} \a \overline{\a}^{\beta}}\right)^{\frac1{\b}},
\left(2\frac{\sqrt{1+ \|b\|^2_{\ell_2}}}{{C_1} \a_{n_{\a}} \a \overline{\a}^{\beta}}\right)^{\frac1{\b-1}}\right).
\end{eqnarray*}
Then, setting
$M_{x,G}:=\max \left(\frac{x-x_0}{\overline \a}, \frac{|x|}{\overline \a},\overline  M_{x,G}\right),$ if 
$\|\underline h\|_{\ell_2}+\|\overline h\|_{\ell_2} >{M_{x,G}},$
  \begin{eqnarray}\label{lajta}
EU(V^{x,h}-G) & <  & EU(x-G)
 \end{eqnarray}
so the strategy $0$ performs better than $h$. It follows that 
$\| h\|_{\ell_2} >M_{x,G}$ implies \eqref{lajta} since
$$\| h\|_{\ell_2}=\left(\|\underline h\|^2_{\ell_2} + \|\overline h\|^2_{\ell_2}\right)^{\frac12}  \leq 
\|\underline h\|_{\ell_2} + \|\overline h\|_{\ell_2}.$$
\end{proof}

Now we present our first main result. We establish the existence of an optimizer for the utility maximization problem.
In \cite{ijtaf} this was shown assuming uniformly bounded exponential moments for the $\varepsilon_{i}$.
In \cite{jmaa} the moment condition was weak but it was assumed that all the $\varepsilon_{i}$
take arbitrarily large negative and positive values. Here
we do not need the latter assumption and merely assume \eqref{harom} and \eqref{harom2}.
 
\begin{theorem}\label{csonti}
Assume that  Assumptions \ref{un}, \ref{b},  \ref{AOAfini}, \ref{trois}, \ref{queuneg1}, \ref{queuneg2} and \ref{intG} hold true. Let $x\in \R$. There exists
$h^*\in\mathcal{A}(U,G,x)$ such that
$$
u(G,x)=EU(V^{x,h^*}-G).
$$
\end{theorem}
\begin{proof}
 Let $x \in \R$ and
let $h_n\in\mathcal{A}(U,G,x)$ be a sequence such that
$$
EU(V^{x,h_n}-G)\uparrow u(G,x),\ n\to\infty.
$$
If $\|h_n\|_{\ell_2} >M_{x,G},$ 
 then using Lemma \ref{hborne}, we can replace $h_n$ by $0$ and still have a maximising sequence. So one can assume that $\sup_{n\in\mathbb{N}} \|h_n\|_{\ell_2}\leq M_{x,G}<\infty.$ 
Hence as $\ell_2$ has the Banach-Saks Property, there exists a subsequence $(n_k)_{k\geq 1}$ and some $h^*\in \ell_2$ such that for $\widetilde{h}_n:=\frac{1}{n}\sum_{k=1}^n h_{n_k}$
 $$\|\widetilde{h}_n-h^*\|_{\ell_2}\to 0, \,n\to\infty.$$
Using \eqref{isol}, we get that
\begin{eqnarray*}
E\langle \widetilde{h}_n-h^*,\varepsilon-b\rangle^2
& \leq & \| \widetilde{h}_n-h^*\|^2_{\ell_2}(1+\|b\|^2_{\ell_2})\to 0,
\end{eqnarray*}
when $n\to\infty$. In particular, $\langle \widetilde{h}_n-h^*,\varepsilon-b\rangle\to 0$, $n\to\infty$
in probability. Hence also $U(V^{x,\widetilde{h}_n}-G)\to U(V^{x,h^*}-G)$ in probability by continuity  of $U$.
We claim that the family $U^+(V^{x,\widetilde{h}_n}-G)$, $n\in\mathbb{N}$ is uniformly
integrable.
Indeed, from \eqref{ilfaitbeau}
$$
U^+(V^{x,\widetilde{h}_n}-G)\leq |x_0|+ |V^{x,\widetilde{h}_n}|.
$$
We know that $\sup_{n\in\mathbb{N}} \|\widetilde{h}_n\|_{\ell_2}\leq M_{x,G}<\infty$.
Hence from Assumption \ref{trois} (see Lemma \ref{miki} and Remark \ref{remui}), we get that  $\{U^+(V^{x,\widetilde{h}_n}-G), \, h_n\in \ell_2, \, \|\widetilde{h}_n\|_{\ell_2} \leq M_{x,G}\}$ is uniformly integrable. Fatou's lemma used for $-U^-$ implies that
$$E\left(-U^-(V^{x,h^*}-G)\right)\geq \limsup_{n\to\infty}E\left(-U^-(V^{x,\widetilde{h}_n}-G)\right),$$
and uniform integrability guarantees that
$$\lim_{n\to\infty}E\left(U^+(V^{x,\widetilde{h}_n}-G)\right)=
E\left(U^+(V^{x,{h}^*}-G)\right).$$
Thus, by concavity of $U$
$$
EU(V^{x,h^*}-G)\geq \limsup_{n\to\infty}EU(V^{x,\widetilde{h}_n}-G)\geq \lim_{n\to\infty}EU(V^{x,{h}_n}-G)=u(G,x),
$$
and the proof will be finished  as soon as we show $h^* \in \mathcal{A}(U,G,x)$. 
From Assumption \ref{queuneg2} and Lemma \ref{miki},
\begin{eqnarray}
\nonumber
EU^-(V^{x,\widetilde{h}_n}-G) &\leq &  C_3E|V^{x,\widetilde{h}_n}-G|^{\gamma} +C_4 \\
\nonumber
&\leq&  C_3\left(2^{\gamma-1}(|x|^{\gamma}+ E|<\widetilde{h}_n, \e-b>|^{\gamma})\right) +C_4\\
\label{admiss}
&\leq & C_3\left(2^{\gamma-1}\left(|x|^{\gamma}+C_{\gamma} M_{x,G}^{\gamma}
\left(1+\|b \|_{\ell_2}^{\gamma}\right)\right)\right) +C_4=:K.
\end{eqnarray}
Fatou's lemma used for $U^-$ implies that
\begin{eqnarray*} 
E\left(U^-(V^{x,h^*}-G)\right) & \leq & \liminf_{n\to\infty}E\left(U^-(V^{x,\widetilde{h}_n}-G)\right)  \leq K.
\end{eqnarray*}
\end{proof}

We consider now the problem of optimization in the small market $n$ with only the random sources $(\e_i)_{1 \leq i \leq n}.$ Let 
$$\mathcal{A}_n(U,G,x):=\left\{ h \in \ell_2,\; h_i=0, \, \forall i \geq n+1, \, E U^{-}(V^{x,h}-G)<+\infty \right\}.$$
Note that $\mathcal{A}_n(U,G,x) \subset \mathcal{A}_{n+1}(U,G,x) \subset \ldots \subset \mathcal{A}(U,G,x).$
We set for $n\in\mathbb{N}$
\begin{eqnarray}
\label{maxupetit}
u_n(G,x):=\sup_{h\in\mathcal{A}_n(U,G,x)}EU(V^{x,h}-G).
\end{eqnarray}

Now we arrive at the principal message of our paper: optimization problems in the small markets
behave consistently with those on the big market, in a natural way.

\begin{theorem}\label{matural}
\label{zut} Assume that Assumptions \ref{un}, \ref{b},  \ref{AOAfini}, \ref{trois},  \ref{queuneg1}, \ref{queuneg2} and \ref{intG} hold true. Then for each $x\in \R$, we have $u_n(G,x)\uparrow u(G,x)$, $n\to\infty$. \\
Let $h_n^*$ be an optimal solution for \eqref{maxupetit} \footnote{which exists by the argument of Theorem \ref{csonti}.}. Then there exists a subsequence $(n_k)_{k\geq 1}$ and some $\widehat h\in \ell_2,$ optimal solution of \eqref{gnon}, such that for $\widehat{h}_n:=\frac{1}{n}\sum_{k=1}^n h^*_{n_k}$, 
$$\|\widehat{h}_n-\hat h\|_{\ell_2}\to 0, \,n\to\infty.$$
\end{theorem}
\begin{proof} The sequence
$u_n(G,x)$, $n\in\mathbb{N}$ is clearly non-decreasing and it is bounded from above by $u(G,x)$.
Let $\bar{h}_n:=(\widetilde{h}_0,\ldots,\widetilde{h}_n,0,\ldots)$, $n\in\mathbb{N}$ where $\widetilde{h}$ is the
optimizer constructed in Theorem \ref{csonti}. Using \eqref{isol} and $\widetilde{h} \in \ell_2$, we have $$
E\langle \bar{h}_n-\widetilde{h},\varepsilon-b\rangle^2\to 0,\ n\to\infty
$$
hence also $\langle \bar{h}_n,\varepsilon-b\rangle\to \langle \widetilde{h},\varepsilon-b\rangle$, $n \to \infty$ in probability.
The Fatou-lemma for $U^+$ shows that
$$
EU^+(V^{x,\widetilde{h}}-G)\leq \liminf_{n\to\infty}EU^+(V^{x,\bar{h}_n}-G).
$$
Now we show that the family $U^-(V^{x,\bar{h}_n}-G)$, $n\in\mathbb{N}$ is uniformly
integrable.  Assumption  \ref{queuneg2} implies that 
\begin{eqnarray*}
U^-(V^{x,\bar{h}_n}-G) &\leq &  C_3|V^{x,\bar{h}_n}-G|^{\gamma} +C_4 \\
&\leq&  C_3\left(2^{\gamma-1}(|x|^{\gamma}+ |<\bar{h}_n, \e-b>|^{\gamma})\right) +C_4.
\end{eqnarray*}
As $\widetilde{h}$ is optimal, $\|\bar{h}_n\|_{\ell_2}\leq\|\widetilde{h}\|_{\ell_2} \leq  M_{x,G}$ (see Lemma \ref{hborne}) and
as in Remark \ref{remui}, $U^-(V^{x,\bar{h}_n}-G)$, $n\in\mathbb{N}$ is uniformly
integrable. We also get as 
in \eqref{admiss}  that $$EU^-(V^{x,\bar{h}_n}-G)  \leq K$$ 
and $\bar{h}_n \in \mathcal{A}_n(G,U,x)$ follows. Uniform integrability implies that
$$
EU^-(V^{x,\widetilde{h}}-G)= \lim_{n\to\infty}EU^-(V^{x,\bar{h}_n}-G).
$$
It follows that
$$
u(G,x)=EU(V^{x,\widetilde{h}}-G)\leq \liminf_{n\to\infty}EU(V^{x,\bar{h}_n}-G)\leq \lim_{n\to\infty}u_n(G,x) \leq u(G,x).
$$
Let $h_n^*\in\mathcal{A}_n(U,G,x)$ be an optimal solution for \eqref{maxupetit}.  
Using Lemma \ref{hborne}, $\|{h}^*_n\|_{\ell_2} \leq M_{x,G}$. 
We proceed as in the proof of Theorem \ref{csonti}. By the Banach-Saks Property, 
there exists a subsequence $(n_k)_{k\geq 1}$ such that for 
$\widehat{h}_n:=\frac{1}{n}\sum_{k=1}^n h^*_{n_k},$ 
 $$\|\widehat{h}_n-\widehat h\|_{\ell_2}\to 0, \,n\to\infty$$ 
for some $\widehat h \in\ell_2$. The arguments of the proof of Theorem \ref{csonti} apply verbatim and show that $\widehat{h}$ is
an optimizer for the utility maximization problem \eqref{gnon} in the large market.
\end{proof}

\begin{remark}
When $U$ is strictly concave then the optimizer is unique and hence $h^{*}$ of Theorem \ref{csonti} equals $\widehat{h}$ of Theorem \ref{matural}. 
\end{remark}

The corollary below addressees the problem of convergence of the reservation prices $p_{n}$, $p$. These latter were introduced in
\cite{hodges-neuberger}. 

\begin{corollary}  
\label{prixut}
Assume that Assumptions \ref{un}, \ref{b},  \ref{AOAfini}, \ref{trois},  \ref{queuneg1}, \ref{queuneg2} and \ref{intG} hold true.
The reservation price $p_n$ (resp. $p$) of  $G$ in the market with the random sources $(\e_i)_{1 \leq i \leq n}$  (resp. with $(\e_i)_{ i \geq 1}$) 
is defined as a solution of 
\begin{eqnarray*}
u_n(G,x+p_n) =  u_n(0,x)  &\mbox{ and } &
u(G,x+p)  =  u(0,x).
\end{eqnarray*}
These quantities are well-defined and 
we have $p_n\to p$, $n\to\infty$.
\end{corollary}
\begin{proof} We justify the definition of $p$, the case of $p_{n}$ being completely analogous.
We show that the set $\{u(G,x),  \,x\in\mathbb{R}\}$ is the same as $\{u(0,x),  \,x\in\mathbb{R}\}$.

We claim that $u(G,x)$, $u(0,x)$ are finite for all $x$.
Indeed, Assumption \ref{intG}, Lemmata \ref{toutva} and \ref{hborne} imply that $-\infty<u(G,x)\leq u(0,x)<\infty$.
{}
As $u$ is monotone, furthermore it is concave and thus continuous on its effective domain, it suffices to show that 
\begin{equation}\label{shovel}
u(G,-\infty)=u(0,-\infty)=-\infty,\ u(G,\infty)=u(0,\infty)=U(\infty)
\end{equation}
and that $u(G,x),u(0,x)<U(\infty)$ for all $x$ because in this case 
$\{u(G,x),  \,x\in\mathbb{R}\}=\{u(0,x),  \,x\in\mathbb{R}\}=(-\infty,U(\infty))$. 

We first concentrate on the latter claim. If $U(\infty)=\infty$ then this is obvious. Otherwise 
denote $h'$, $h''$ the strategies attaining $u(0,x)$, $u(G,x)$, respectively. Then, by the strictly increasing
property of $U$, we have 
\begin{equation}\label{kelmajd}
u(0,x)=EU(x+\langle h',\varepsilon-b\rangle)<EU(\infty)=U(\infty)	
\end{equation} and 
\begin{equation*}\label{kelmajd2}
u(G,x)=EU(x+\langle h'',\varepsilon-b\rangle-G)<EU(\infty)=U(\infty).
\end{equation*}

Now we turn to showing \eqref{shovel}.
It is clear that $u(G,\infty),u(0,\infty)\leq U(\infty)$ and 
\begin{equation}\label{romania}
u(0,\infty)=\lim_{x\to\infty}u(0,x)\geq \lim_{x\to\infty}U(x)=U(\infty).{}
\end{equation}
Assumption \ref{intG} and Fatou's lemma also imply that $$
u(G,\infty) \geq \liminf_{x\to\infty}u(G,x)\geq\liminf_{x\to\infty}EU(x-G)\geq U(\infty).
$$

Since $u(G,x)\leq u(0,x)$, it is enough to establish $\lim_{x\to -\infty}u(0,x)=-\infty$.
By concavity, this is clearly the case if $u(0,\cdot)$ is not the constant function. But if $u(0,\cdot)=c$ then 
we would necessarily have $c\geq U(\infty)$ by \eqref{romania} which contradicts \eqref{kelmajd}.

We now turn to proving convergence. Arguing by contradiction let us assume that, along a subsequence (which we continue to denote by $n$),
one has $p_n\to \underline{p}$ for some $\underline{p}<p$ (the case of a limit $\overline{p}>p$ is
analogous). It follows that there is $N$ such that, for $n\geq N$, $p_n<(p+\underline{p})/2<p$.
Using Theorem \ref{csonti}, let $h^{\dagger}\in \mathcal{A}(G,U,x+(p+\underline{p})/2) \subset \mathcal{A}(G,U,x+p)$ satisfy
$$
u(G,x+(p+\underline{p})/2)=EU(x+(p+\underline{p})/2+\langle h^{\dagger},\varepsilon-b\rangle -G).
$$
Then, the definition of the reservation prices and Theorem \ref{zut} imply that 
\begin{eqnarray*}
& & \limsup_{n\to\infty}u_n(G,x+p_n)\leq \limsup_{n\to\infty}u_n(G,x+(p+\underline{p})/2)\\
& =& u(G,x+(p+\underline{p})/2)= EU(x+(p+\underline{p})/2+\langle h^{\dagger},\varepsilon-b\rangle-G)\\
&<& EU(x+p+\langle h^{\dagger},\varepsilon-b\rangle-G)\leq
u(G,x+p)\\
&=& u(0,x)=\lim_{n\to\infty}u_n(0,x) = \lim_{n\to\infty}u_n(G,x+p_n),
\end{eqnarray*}
a gross contradiction.
\end{proof}

\subsection*{Acknowledgments}

M.R. was supported by the National Research, Development and Innovation Office, Hungary [Grant KH 126505] and
by the ``Lend\"ulet'' programme of the Hungarian Academy of Sciences [Grant LP 2015-6].

\end{document}